\documentclass{svproc}
\usepackage[colorlinks=true, allcolors=blue]{hyperref}
\usepackage[english]{babel}
\usepackage[usenames]{color}
\usepackage{adjustbox}
\usepackage{algorithm}
\usepackage{algpseudocode}
\usepackage{amsfonts}
\usepackage{amsmath}
\usepackage{amssymb}
\usepackage{graphicx}
\usepackage{latexsym}
\usepackage{lscape}
\usepackage{nicematrix}
\usepackage{tikz}
\usepackage{todonotes}
\usepackage{diagbox}

\usepackage{todonotes}

\newcommand{\ignore}[1]{}

\newcommand{\beqn}{\begin{eqnarray*}}
\newcommand{\eeqn}{\end{eqnarray*}}

\newcommand{\vc}[1]{\mathbf{#1}}

\newcommand{\EE}{{\mathcal E}}

\newcommand{\GG}{{\mathcal G}}

\newcommand{\NN}{{\mathcal N}}

\newcommand{\XX}{{\mathcal X}}

\renewcommand{\epsilon}{\varepsilon}
{\begingroup

  \begin{enumerate}}
  {\end{enumerate}\endgroup}

\newcommand{\beq}{\begin{eqnarray*}}
\newcommand{\eeq}{\end{eqnarray*}}

\newcommand{\lra}{\leftrightarrow}

\usetikzlibrary{arrows.meta,calc}
\tikzset{>={Latex[width=1.5mm,length=1.5mm]}}

\newcommand{\RR}{\mathbf R}

\usepackage{authblk}

\begin{document}

\title{Recognizing Distance-Count Matrices is Difficult}
\author{Paolo Boldi\inst{1} \and
	Flavio Furia\inst{1} \and
	Chiara Prezioso\inst{1} \and
	Ian Stewart\inst{2}
}
\institute{
  \begin{minipage}[t]{0.48\textwidth}
    \centering
    $^1$ Computer Science Dept. \\
    Università degli Studi \\
    Milano, Italy
  \end{minipage}
  \hfill
  \begin{minipage}[t]{0.48\textwidth}
    \centering
    $^2$ Mathematics Institute \\
    University of Warwick \\
    Coventry CV4 7AL, UK
  \end{minipage}
}
\authorrunning{Paolo Boldi et al.}



\maketitle

\begin{abstract}
Axiomatization of centrality measures often involves proving that something cannot hold by providing a counterexample (i.e., a graph for which that specific centrality index fails to have a given property). In the context of geometric centralities, building such counterexamples requires constructing a graph with specific distance counts between nodes, as expressed by its distance-count matrix.
We prove that deciding whether a matrix is the distance-count matrix of a graph is strongly NP-complete. This negative result implies that a brute-force approach to building this kind of counterexample is out of question, and cleverer approaches are required.
\end{abstract}

\section{Introduction and Related Work}

The distance-count matrix (DCM) of a graph is a matrix whose $n$ rows correspond to the vertices, and where the $k$-th column for vertex $v$ contains the number of vertices whose (shortest-path) distance to $v$ is equal to $k$. 
The DCM contains a lot of information about the graph itself, and it is a natural object to study, especially in the context of social network analysis: a large family of centrality measures, called geometric centralities~\cite{BoVAC}, are those that can be expressed as a function of the DCM of a graph. This family includes degree centrality, closeness~\cite{BavCPTOG}, Lin centrality~\cite{LinFSR}, harmonic centrality~\cite{BoVAC}, and many others. The subfamily of linear geometric centralities alone was also studied in the pioneering works in Kishi \emph{et al.}~\cite{KisOCFG,KiTTCFDG}, and more recently in~\cite{SkSADBC,BFPLC}.

The DCM implicitly contains other information, such as the eccentricity of all vertices (and their 
distribution), the diameter and effective diameter~\cite{LKFGEDSD} of the graph, the distance distribution~\cite{BBRFDS}, and the graph Wiener index~\cite{wiener1947structural,rouvray2002rich}. The problem of computing or approximating the DCM is therefore extremely important in practice and challenging, especially for large graphs~\cite{boldi2013core}. While building the DCM of a given graph is relatively easy, the main result we prove is that deciding whether a matrix is a DCM is strongly NP-complete. This negative result is especially relevant when trying to build counterexamples to specific properties of geometric centralities: in such a situation, it may be possible to find a candidate DCM that works as a counterexample, but then the problem is to determine whether that candidate matrix \emph{is} the actual DCM of some graph. Our result implies that this road is basically impossible, so we must explore a different technique for finding a counterexample.

Determining whether a given sequence of integers is a graphical degree sequence, that is, the degree sequence of an undirected graph is a well-known problem in graph theory: among the first results about this problem are the celebrated Erd\H{o}s-Gallai theorem~\cite{erdHos1960grafok}, characterizing graphical degree sequences, and the Havel-Hakimi algorithm~\cite{havel1955remark,hakimi1962realizability}, which gives a constructive way to check in polynomial time whether a sequence is a graphical degree sequence, and provides, in the positive case, a possible realization of the sequence as a graph. Later, the problem was also studied in the context of directed graphs~\cite{kleitman1973algorithms}, where instead of the degree sequence one considers the in-degree and out-degree sequences. 

Another quite natural generalization of the problem is to take a list of pairs of integers $(d_i,d_i')$ and determine whether this is a second-order degree sequence of a graph. Such a sequence contains, for every vertex $i$, the number of vertices $d_i$ at distance $1$ (i.e., its degree) and the number of vertices $d_i'$ at distance $2$. This version of the problem, however, is NP-complete~\cite{ErMNASLDSPE}. 
In this paper, we consider an even more general version of the problem, where our input is the whole distance-count matrix of a graph.

The intuition that recognizing DCMs may itself be NP-complete comes from the fact that the first two columns of the DCM of a graph are precisely its second-order degree sequence: in the light of~\cite{ErMNASLDSPE} we can expect that the problem of recognizing DCMs is at least as hard. But the DCM contains much more information than the second-order degree sequence: it contains the number of vertices at distance $k$ for every $k$, and therefore it is a complete description of the distances in the graph. So, we may think that this additional knowledge could make the problem easier.
Formally, there is no easy polynomial reduction between the two problems, in either direction.

The rows of the DCM have been considered in the graph-theoretical literature under the name of distance degree sequences (or dds)~\cite[Section 9.2]{BH90}, and were studied in some special scenarios (e.g., to determine which graphs have vertices with pairwise distinct dds, and which have the same dds for all vertices~\cite{huilgol2014distance}), but not in the general case.

\section{Notation}

A \emph{graph} $\GG=(\NN_\GG,\EE_\GG)$ is given~\cite{BolMGT} by a finite set of nodes $\NN_\GG$ and a set of arcs $\EE_\GG\subseteq \NN_\GG \times \NN_\GG$. (For this and similar notations, the subscript $\GG$ is dropped whenever it is clear from the context.)
Without loss of generality we let $\NN=[n]=\{0,1,\dots,n-1\}$ where $n$ is the number of nodes.
We write $x\to y$ to mean that $x,y\in \NN$ and $(x,y)\in \EE$; we say that $x$ is a \emph{predecessor} of $y$, that $y$ is a \emph{successor} of $x$, and that $x$ ($y$, respectively) is the \emph{tail} (\emph{head}, respectively) of the arc $x \to y$.
A graph is \emph{undirected} iff $x\to y$ implies $y\to x$. For undirected graphs, we write simply $x \lra y$, and say that $x$ and $y$ are \emph{adjacent}.

A \emph{path} of length $k$ from $x$ to $y$ is a sequence $\pi=(x_0,x_1,\dots,x_k)$ of nodes such that $x=x_0$, $y=x_k$ and $x_i \to x_{i+1}$ for all $i\in [k]$; the arcs $x_i \to x_{i+1}$ are said to \emph{belong} to the path. We say that a path is \emph{simple} if it does not contain any node twice. 

The \emph{(shortest path) distance} from $x$ to $y$ in $\GG$, denoted by $d_\GG(x,y)$, is the length of a shortest path from $x$ to $y$, or $\infty$ if no path from $x$ to $y$ exists. A \emph{(strongly) connected graph} is one where the distance between every pair of nodes is finite (the adverb ``strongly'' is usually omitted for undirected graphs). 

In a connected undirected graph, the distance function $d(-,-)$ is a metric, that is, for all $x,y,z \in V$ (1) $d(x,y)\geq 0$ and equality holds if and only if $x=y$; (2) $d(x,y)=d(y,x)$ (symmetry); (3) $d(x,y)\leq d(x,z)+d(z,y)$ (triangle inequality). In a strongly connected directed graph symmetry does not hold in general, while the other two properties remain true.

The \emph{(in)-eccentricity} of a node $x$ in $\GG$ is defined as the maximum finite $d(y,x)$ as $y$ ranges over all nodes of $\GG$. The maximum eccentricity of all nodes in $\GG$ is called the \emph{diameter} of $\GG$.

\begin{figure}
    \centering
        \begin{tikzpicture}[scale=0.46, ->, node distance=0.8cm and 0.8cm, every node/.style={draw, circle, minimum width=0.6cm, 
    minimum height=0.4cm}]
        \node (0) at (0,0) {0};
        \node (3) at (2,-2) {3};
        \node (4) at (0,-4) {4};
        \node (6) at (-2,-6) {6};
        \node (7) at (-2,-8) {7};
        \node (5) at (0.9,-7) {5};
        \node (1) at (-3,-10) {1};
        \node (2) at (-1,-10) {2};
    
        \draw (0) to (3);
        \draw (0) to[bend right] (4);
        \draw (3) to (4);
        \draw (4) to[bend right] (0);
        \draw (4) to (6);
        \draw (6) to (7);
        \draw (7) to (1);
        \draw (7) to[bend right] (2);
        \draw (2) to[bend right] (7);
        \draw (1) to[bend left] (6);
        \draw (6) to (5);
        \draw (5) to (3);
        \draw (3) to[bend left=40] (2);
        \draw (2) to (4);
        \draw (1) to[bend left=40] (0);
    \end{tikzpicture}
    \caption{A strongly connected graph $\GG$.\label{fig:graph}}
\end{figure}

\section{Distance-count matrices (DCM)}

We start by giving some definitions:

\begin{definition}[DCM and CDCM]
Given a graph $\GG$, a node $x$ and a natural number $k$, define 
\begin{center}
    \begin{tabular}{lcl}
    $N_{\GG,k}(x) = \{y \in \NN \mid d(y,x)=k\}$ & $\quad$ & $n_{\GG,k}(x) = \left|N_{\GG,k}(x)\right|$\\
    $M_{\GG,k}(x) = \{y \in \NN \mid d(y,x)\leq k\}$ & $\quad$ & $m_{\GG,k}(x) = \left|M_{\GG,k}(x)\right|$.
    \end{tabular}
\end{center}

The \emph{distance-count matrix (DCM)} of $\GG$ is the matrix $N=N_\GG \in \RR^{n \times n}$ such that $n_{i,k}=n_{\GG,k}(i)$, i.e., the number of nodes at distance exactly $k$ to $i$.
The \emph{cumulative distance-count matrix (CDCM)} of $\GG$ is the matrix $M=M_\GG \in \RR^{n \times n}$ such that $m_{i,k}=m_{\GG,k}(i)$, i.e., the number of nodes at distance at most $k$ to $i$.
(It is convenient to assume that vector and matrix elements are indexed starting from zero.)
\end{definition}

Observe that for all $k$,
$M_k(x) \setminus M_{k-1}(x) = N_k(x)$
where we conveniently assume that $M_{-1}(x)=\emptyset$.
The definitions describe the set of nodes that are found by moving further and further away from a node. We can think of it as a circular wave that starts from node $x$: as $k$ increases we find nodes that are more and more distant from $x$. Eventually, after $k$ has reached the eccentricity of $x$, no more nodes are found. We can examine either the set (or count) of nodes at any given distance from the centre (the functions $N$ and $n$), or (equivalently) the cumulative set or count of nodes found along the way (the functions $M$ and $m$).  Also observe that we are computing distances \emph{to} $x$, not \emph{from} $x$: this fact is a matter of convention, all our results hold also if we consider the distance from $x$ to $y$ instead of the distance from $y$ to $x$; in the undirected case, the two are the same.

\begin{figure}
    \centering
    \begin{tabular}{cc}
    $N_\GG=
        \begin{pNiceMatrix}
        1 & 2 & 3 & 2 & 0 & 0 & 0 & 0\\
        \Block[fill=black!25]{1-8}{}
        1 & 1 & 2 & 2 & 2 & 0 & 0 & 0\\
        1 & 2 & 3 & 2 & 0 & 0 & 0 & 0\\
        \Block[fill=black!25]{1-8}{}
        1 & 2 & 3 & 2 & 0 & 0 & 0 & 0\\
        1 & 3 & 3 & 1 & 0 & 0 & 0 & 0\\
        1 & 1 & 2 & 4 & 0 & 0 & 0 & 0\\
        1 & 2 & 4 & 1 & 0 & 0 & 0 & 0\\
        1 & 2 & 3 & 2 & 0 & 0 & 0 & 0            
        \end{pNiceMatrix}
    $
    &
    $ M_\GG=
        \begin{pNiceMatrix}
        1 & 3 & 6 & 8 & 8 & 8 & 8 & 8\\
        1 & 2 & 4 & 6 & 8 & 8 & 8 & 8\\
        1 & 3 & 6 & 8 & 8 & 8 & 8 & 8\\
        1 & 3 & 6 & 8 & 8 & 8 & 8 & 8\\
        1 & 4 & 7 & 8 & 8 & 8 & 8 & 8\\
        1 & 2 & 4 & 8 & 8 & 8 & 8 & 8\\
        1 & 3 & 7 & 8 & 8 & 8 & 8 & 8\\
        1 & 3 & 6 & 8 & 8 & 8 & 8 & 8            
        \end{pNiceMatrix}
    $ 
    \end{tabular}
    \caption{The DCM $N_\GG$ and CDCM $M_\GG$ for the graph $\GG$ of Figure~\ref{fig:graph}. The highlighted rows correspond to $n_{\GG,-}(1)$ and $n_{\GG,-}(3)$. \label{fig:matrix}}
\end{figure}

\subsection{Graphical sequences for undirected graphs}

The standard reference about distances in graphs is \cite{BH90}: this book includes material on DCMs under different terminology for the undirected case.
The {\em distance degree sequence} (or $dds$) of node $v$ of an undirected graph (see \cite[Section 9.2]{BH90})
is the vector 
\[
dds(v) = (d_0(v), d_1(v), \ldots, d_{e(v)}(v))
\]
where $d_i(v)$ is the number of nodes $w$ such that $d(v,w) = i$,
and $e(v)$ is the {\em eccentricity} of $v$, which is the maximum
value of $d(v,w)$ for $w \in \GG$. Thus, $dds(v)$ is row $v$
of the DCM, truncated to remove the final zeros.

The {\em distance degree sequence (or dds)} $dds(\GG)$ of the graph $\GG$ is the list of all $dds$'s of its nodes, listed including multiplicity
(see \cite[Section 9.2]{BH90}).
This is essentially the DCM with zero entries removed.

Also of interest is the {\em degree sequence} of an undirected graph $\GG$ (see \cite[Section 1.1]{BH90}),
which is the list of degrees of nodes arranged in nonincreasing order.
Up to permutation of the nodes, this is column $1$ of the DCM of $\GG$.

Not all sequences of positive integers can be a degree sequence.
A {\em graphical degree sequence} is a degree sequence 
that can be realised by an undirected graph.
A complete characterization of graphical degree sequences for undirected graphs is given by the Erd\H{o}s-Gallai theorem:
\begin{theorem}[\cite{erdHos1960grafok}]
    Consider a sequence $d_1 \geq d_2 \geq \dots \geq d_p$ of $p$ positive integers. This sequence is a graphical degree sequence if and only if its sum is even and for every $n=1,\dots,p-1$
    \begin{equation}
        \label{equ:eg}
        \sum_{k=1}^n d_k \leq n(n-1) + \sum_{k=n+1}^p \min(n, d_k).
    \end{equation}
\end{theorem}

The Havel-Hakimi algorithm~\cite{havel1955remark,hakimi1962realizability} provides a constructive way to check in polynomial time if a sequence is a graphical degree sequence. The algorithm is essentially described by the following~\cite[Theorem 9.1]{BH90}:

\begin{theorem}
\label{T:gds}
The sequence $D= (d_1, \ldots, d_p)$ with $p-1 \geq d_1 \geq d_2 \geq \cdots \geq d_p$
for positive integers $d_i$ is a graphical degree sequence if and only if
\[
D' = (d_2-1, d_3-1, \ldots d_{d_1+1}-1, d_{d_1+2}, \ldots, d_p),
\]
when re-ordered in nonincreasing order, is a graphical degree sequence.
\qed\end{theorem}

The Havel-Hakimi algorithm applies Theorem~\ref{T:gds} inductively: if the process stops at the empty sequence, the original sequence is a graphical degree sequence; if any entry becomes negative before this happens, it is not a graphical degree sequence.

As explained above, column 1 of the DCM of an undirected graph is always a graphical degree sequence. 
This fact is more general: all columns of a DCM are graphical degree sequences; more precisely, column $k$ of the DCM of an undirected graph $\GG$ is the degree sequence of the graph $\GG^k$ having the same vertices as $\GG$ and with an edge $i \leftrightarrow j$ iff $d(i,j) = k$. 

\subsection{Graphical in-degree sequences}

\emph{Mutatis mutandis}, we can define the \emph{in-degree sequence} of a directed graph $\GG$ as the list of in-degrees of nodes arranged in non-increasing order. Up to permutation of the nodes, this is column 1 of the DCM of $\GG$. We can ask whether the algorithm described in Theorem~\ref{T:gds} can be adapted to in-degree sequences. In fact, the result for directed graphs is much simpler:

\begin{theorem}
\label{T:gds-dir}
The sequence $D= (d_1, \ldots, d_p)$ with $p \geq d_1 \geq d_2 \geq \cdots \geq d_p$
for natural numbers $d_i$ is always the in-degree sequence of a directed graph.
\end{theorem}
\begin{proof}
    Just use $\NN=\{1,2,\dots,p\}$ and for every $i \in \NN$, add $d_i$ edges, whose head is node $i$ and whose tails are $d_i \leq p$ distinct but arbitrarily chosen elements of $\NN$.\qed
\end{proof}

A more sophisticated question is whethere a given sequence of pairs is the sequence of in- and out-degrees of a graph:
the Kleitman–Wang Algorithm~\cite{kleitman1973algorithms} extends the Havel-Hakimi Algorithm for this case.

\subsection{Some basic properties of (C)DCMs}

In this brief section we collect some observations about (C)DCMs. Let us start with some trivial ones.

\begin{proposition}
\label{P:DCM2}
In any graph $\GG$, for any node $i$ of $\GG$, and for any $0 \leq p \leq n-1$:

{\rm (a)} $M_0(i) = N_0(i) = \{i\}$.

{\rm (b)} $m_0(i) = n_0(i) = 1$.

{\rm (c)} $m_1(i) = \nu(i)+ 1$ and $n_1(i) = \nu(i)$ where $\nu(i)$ is the in-degree (number of predecessors) of $i$.

{\rm (d)} $M_p(i) \subseteq M_{p+1}(i)$.

{\rm (e)} $m_p(i) \leq m_{p+1}(i)$.

{\rm (f)} $n_p(i) = m_p(i) - m_{p-1}(i)$.

{\rm (g)} $m_p(i) = n_0(i)+n_1(i) + \cdots + n_p(i)$.

{\rm (h)} If $M$ is a DCM (resp. CDCM), so is any matrix obtained by permuting the rows of $M$.

{\rm (i)} $m_p(n-1) = n$ if and only if $d(j,i)$ is finite for all nodes $j$; in particular, this is always true if $\GG$ is strongly connected.
\qed\end{proposition}

Observe that the (C)DCM of a graph is uniquely defined only up to row permutation, but by (h) we can use lexicographic order to define a unique (C)DCM of a graph, that we shall call the \emph{canonical} (C)DCM of $\GG$.

We can extend the function $M_p(-)$ to subsets: for $\XX \subseteq \NN$, define the {\em $p$-neighbourhood} of $\XX$ to be
\[
M_p(\XX) = \bigcup_{i \in \XX} M_p(i)
\]
Clearly
\begin{equation}
\label{E:MpMq}
M_p(M_q(\XX)) = M_{p+q}(\XX).
\end{equation}

\begin{proposition}
    For any graph $\GG$ and node $i$:
    \label{P:DCM3}
    
    {\rm(a)} If $m_{p+1}(i) = m_p(i)$ then $m_{p+r}(i) = m_p(i)$ for $0 \leq r \leq n-p$.
    
    {\rm(b)} If $\GG$ is strongly connected and $m_{p+1}(i) = m_p(i)$ then $m_p(i) = n$ and $m_{p+r}(i) = n$
    for all $0 \leq r \leq n-p-1$.
\end{proposition}
\begin{proof}
    (a) If $m_p(i+1) = m_p(i)$ then case~(d) of Proposition~\ref{P:DCM2} implies that
    $M_{p+1}(i) = M_p(i)$. We use induction on $r$. The result holds for $r=1$.
    for the induction step:
    \[
    M_{p+(r+1)}(i) = M_1(M_{p+r}(i)) =  M_1(M_p(i)) = M_{p+1}(i)=M_p(i).
    \]
    
    (b) This follows from (a), since $m_{n-1}(i) = n$ for any $i$
    by connectedness of $\GG$.
\end{proof}
    
\begin{corollary}
    \label{C:increasing}
    For any graph $\GG$ and node $i$ the sequence $[m_p(i)]$ for $0 \leq p \leq n-1$ (which is the $i$-th row of the CDCM)
    is monotonic strictly increasing until it reaches a certain value $k$, at which point all subsequent entries equal $k$. The value $k$ is exactly the number of nodes $j$ such that $d(j,i)$ is finite. If $\GG$ is strongly connected, $k=n$.
\end{corollary}

    \begin{definition}[Good sequence]
    \label{D:good}
    For a given $n$, a sequence $(a_0, a_1, \ldots, a_{n-1})$ of positive integers
    is {\em good} if $a_0 =1$ and there exists $j$ and $k\leq n$ such that
    $a_0 < a_1 < \cdots < a_j = k$ 
    and $a_r=k$ for all $r \geq j$. If further $k=n$, we say that the sequence is \emph{very good}.
    \end{definition}
    
    Corollary \ref{C:increasing} states that every row of a CDCM is good. If furthermore the graph is strongly connected, its rows are very good.
    We observe:
    
    \begin{theorem}
    \label{T:all_good}
    Every good sequence occurs in the CDCM of some graph (in particular, some undirected graph).
    \end{theorem}
    
    \begin{proof}
    Let the good sequence be $(a_0, a_1, \ldots, a_{n-1})$, and define $b_i = a_i-a_{i-1} \geq 0$ for all $i>0$, and $b_0=1$.
    Define $\GG$ to be a tree with root node $0$. Node $0$ is connected to nodes $1, \ldots, b_1$. One of those nodes connects to nodes $b_1+1, \ldots, b_1+b_2$, and so on inductively. Then the $0$-th row of $N_\GG$ is $(b_0, b_1, \ldots b_{n-1})$, so the $0$-th row of $M_\GG$ is $(a_0, a_1, \ldots a_{n-1})$. See Figure \ref{F:tree_ex}.\qed
    \end{proof}
    
    \begin{figure}[htb]
    \centerline{%
    \includegraphics[width=.3\textwidth]{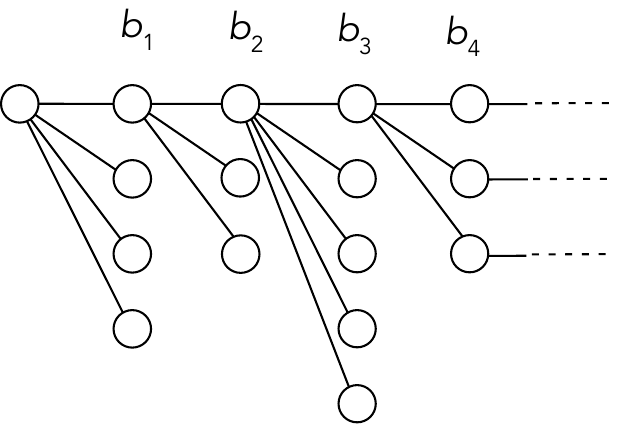}
    }
    \caption{Construction for a graph with a given good row.}
    \label{F:tree_ex}
    \end{figure}
 
\medskip
We end this section with some necessary conditions for a matrix to be a CDCM, in terms of inequalities. We already know:

\begin{theorem}
\label{T:necc_cond1}
For the CDCM of a strongly connected graph:

{\rm (a)} The column $0$ entry in any row is $1$.

{\rm (b)} The column $1$ entry in row $i$ is $m_1(i)=1+\nu(i)$. 
That is, the CDCM (hence also DCM) tells us the in-degree of node $i$.

{\rm (c)} Each row is very good.
\qed\end{theorem}

We now prove two further inequalities.
\begin{theorem}
\label{T:necc_cond2}
Let $M = [m_j(i)]_{i,j}$ be the CDCM of a graph $\GG$. 
Then for any $i \in \NN$, with $\nu = \nu(i) = m_1(i)-1$,
there exist distinct $j_1, \ldots, j_\nu \in \NN \setminus \{i\}$
such that:
\begin{equation}
\label{E:necc_ineq_1}
m_p(i) \geq \max_{k=1}^\nu m_{p-1}(j_k)
\end{equation}
and
\begin{equation}
\label{E:necc_ineq_2}
m_p(i) \leq \sum_{k=1}^\nu m_{p-1}(j_k)
\end{equation}
for all $p=2,\dots,n-1$. 
\end{theorem}
\begin{proof}
Let $j_1,\dots,j_\nu$ be the predecessors of $i$ in $\GG$. To prove \eqref{E:necc_ineq_1}, observe that since 
$j_k\to i$, we have $M_{p-1}(j_k) \subseteq M_p(i)$.
Therefore $m_p(i) \geq m_{p-1}(j_k)$. This holds for all
$k$ such that $1 \leq k \leq \nu$, and this proves \eqref{E:necc_ineq_1}.

To prove \eqref{E:necc_ineq_2}, observe that
by \eqref{E:MpMq},
\[
M_{p}(i) = M_{p-1}(M_1(i))=\bigcup_{j \to i} M_{p-1}(j).
\]
Therefore 
\[
m_p(i) \leq \sum_{j\to i} m_{p-1}(j_k),
\]
which is \eqref{E:necc_ineq_2}.\qed
\end{proof}

That is, the vector of entries of row $i$ after the first is bounded below, termwise, by the maximum of $\nu$ other distinct rows shifted one place to the right.
It is bounded above, termwise, by the sum of the same $\nu$ distinct rows, again shifted one place to the right.
Of course, we do not know which $\nu$ rows to choose, but one
such set of rows must exist. If no such set exists, the matrix
under consideration cannot be a CDCM.

\section{NP-Completeness of (C)DCM recognizability}

In this section, we show that the problem of deciding whether a matrix is a (C)DCM is (strongly) NP-complete.
We use a reduction from a special instance of the \emph{three-partition problem} (TPP)~\cite{GaJCI}, which is a well-known strongly NP-complete problem. 

As already discussed, Erd\H{o}s and Mikl\'os proved in~\cite{ErMNASLDSPE} (in our terminology) that deciding whether two sequences of integers are the columns of indices 1 and 2 of the DCM of a simple undirected graph $\GG$ (the so-called second-order ds problem) is also strongly NP-complete. 
Their proof uses the \emph{basket filling problem}; more precisely, they prove that the basket filling problem is reducible to the second-order ds problem, while the \emph{three-partition problem} is reducible to the second-order ds problem for bipartite graphs. They also
show that the three-partition problem can be reduced in polynomial time to the
basket filling problem. 

As we already mentioned, their reductions do not immediately apply to the (C)DCM recognizability problem, though. A reduction from the second-order ds problem to the (C)DCM recognizability problem would require that we are able, from a given second-order ds instance, to construct a matrix that is the (C)DCM of some graph. It is interesting to observe that our construction also uses the three-partition problem, but in a simpler way than in Erd\H{o}s and Mikl\'os's reduction.

\subsection{The three-partition problem}

Informally, the three-partition problem (shortened hereafter as TPP) consists in deciding if it is possible to partition a list of natural numbers into groups of exactly three integers each, all having the same sum.

This problem is very hard: it is strongly NP-complete (i.e., NP-complete even when the input sequence is provided in unary), and it remains so in many restricted cases. We consider a very mild variant, where we assume that all integers are distinct and separated from one another by a certain threshold.
 
\begin{definition}[TPP instance]
    A \emph{TPP instance} is a sequence of natural numbers $\langle a_1, \dots, a_{3m}\rangle$ such that $a_1\geq a_2\geq \dots \geq a_{3m}$ and $t/4 < a_i < t/2$ for all $i=1,\dots,3m$,
    where $t=(\sum_{i=1}^{3m} a_i)/m$.
\end{definition}

Formally, the three-partition problem is:
\begin{problem}[Three-partition problem (TPP)]
    \label{prob:TPP}
    Given a TPP instance $\langle a_1, \dots, a_{3m}\rangle$, decide whether there exist subsets $D_1,\dots, D_m \subseteq \{1,\dots,3m\}$ such that:
        
        -- the $D_j$'s are pairwise disjoint;
        
        -- $|D_j|=3$ for all $j=1,\dots,m$;
        
        -- $\sum_{i \in D_j} a_i = t$ for all $j=1,\dots,m$,
    
   \noindent where $t=(\sum_{i=1}^{3m} a_i)/m$. If this partition exists, the instance is \emph{positive}, otherwise it is \emph{negative}.
\end{problem}


\begin{theorem}
    \label{thm:tppnpc}
    Problem~\ref{prob:TPP} (TPP) is strongly NP-complete (i.e., it is NP-complete even if the input sequence is assumed to be encoded in unary); it remains so even under the following further assumptions:

            -- all $a_i$'s are distinct,
            
            -- $t \geq 4$, 
            
            -- for a fixed positive $K$, $a_{3m} \geq K$ and $a_i-a_{i+1} \geq K$ for all $i=1,\dots,3m-1$.
 
\end{theorem}
\begin{proof}
    The problem is strongly NP-complete as shown in~\cite{GaJCI}. It remains so even if the input integers are all distinct, as proved in~\cite{hulett2008multigraph}; the fact that the latter result holds also in the presence of the bounds $t/4<a_i<t/2$ is easy (by adding $2t$ to each entry in the sequence).
    
    If $t=(\sum_{i=1}^{3m} a_i)/m\leq 3$, then all $a_i$'s must be zero or one (because $a_i<t/2 \leq 3/2$), so the instance is negative because it contains only one or two elements (recall that the elements in the sequence are all distinct). 
    
    Hence, the problem remains strongly NP-complete for distinct integers with $t \geq 4$. Note that in this case $a_i>t/4\geq 1$, so all $a_i$ are strictly positive. Multiplying all the entries by $K$ we can assume that $a_{3m} \geq K$ and $a_{i}-a_{i+1} \geq K$ for all $i=1,\dots,3m-1$ (multiplying all elements by a constant does not change the problem, because the sequence obtained after the multiplication is a positive instance of the problem if and only if the original sequence was a positive instance itself).\qed
\end{proof}

\subsection{NP-Completeness of (C)DCM}

Given a TPP instance $\vc a = \langle a_1, \dots, a_{3m}\rangle$, let
\[
    s = \sum_{i=1}^{3m} a_i\qquad
    t = s/m\qquad
    n = 4m + s = (t+4)m.
\]
We define a matrix $M(\vc a) \in \RR^{n \times n}$ as follows:

    -- the first $3m$ rows of $M(\vc a)$ have the form
        $[1, a_i, 1, t-a_i, 2, 0, 0, \dots]$
    for $i=1,\dots,3m$,
    
    -- the following $s$ rows of $M(\vc a)$ have the form
        $[1, 2, t-1, 2, 0, 0, 0, \dots]$,
    
     -- the last $m$ rows of $M(\vc a)$ have the form
        $[1, t, 3, 0, 0, 0, \dots]$.

Intuitively, if $\vc a$ is a positive instance of TPP, it is possible to show that $M(\vc a)$ is the DCM of an undirected graph $\GG(\vc a)$. This fact will be part of the proof of Theorem~\ref{thm:prenpc}, but we want to provide immediately a visual clue to this fact for $\vc a=\langle 9, 7, 6, 5, 2, 1\rangle$. (To keep the graph small, we use a sequence that fails to satisfy the more restrictive conditions of Theorem~\ref{thm:tppnpc}.) The graph (see Figure~\ref{fig:ga}) contains $m$ (in this case, 2) connected components, one for each $D_j$. All components contain three vertices (representing the three integers of the input sequence belonging to $D_j$)
connected to $t$ vertices (this is possible because the sum of the three integers is always $t$), which are further connected to a single vertex (of degree $t$).

\begin{figure}[htbp]
    \centering
        \begin{tikzpicture}[main/.style = {draw, circle}, node distance=0.4cm, scale=.4]
        \node[main] (x1) at (0, 0) {};
        \node[main] (x5) at (0, -2.5) {};
        \node[main] (x9) at (0, -9.5) {};
        \node[main] (u1) at (3, 0) {};
        \node[main] (u2) at (3, -1) {};
        \node[main] (u3) at (3, -2) {};
        \node[main] (u4) at (3, -3) {};
        \node[main] (u5) at (3, -4) {};
        \node[main] (u6) at (3, -5) {};
        \node[main] (u7) at (3, -6) {};
        \node[main] (u8) at (3, -7) {};
        \node[main] (u9) at (3, -8) {};
        \node[main] (u10) at (3, -9) {};
        \node[main] (u11) at (3, -10) {};
        \node[main] (u12) at (3, -11) {};
        \node[main] (u13) at (3, -12) {};
        \node[main] (u14) at (3, -13) {};
        \node[main] (u15) at (3, -14) {};
        \node[main] (z1) at (6,.-7) {};
        \draw (x1) -- (u1);
        \draw (x5) -- (u2);
        \draw (x5) -- (u3);
        \draw (x5) -- (u4);
        \draw (x5) -- (u5);
        \draw (x5) -- (u6);
        \draw (x9) -- (u7);
        \draw (x9) -- (u8);
        \draw (x9) -- (u9);
        \draw (x9) -- (u10);
        \draw (x9) -- (u11);
        \draw (x9) -- (u12);
        \draw (x9) -- (u13);
        \draw (x9) -- (u14);
        \draw (x9) -- (u15);
        \draw (u1) -- (z1);
        \draw (u2) -- (z1);
        \draw (u3) -- (z1);
        \draw (u4) -- (z1);
        \draw (u5) -- (z1);
        \draw (u6) -- (z1);
        \draw (u7) -- (z1);
        \draw (u8) -- (z1);
        \draw (u9) -- (z1);
        \draw (u10) -- (z1);
        \draw (u11) -- (z1);
        \draw (u12) -- (z1);
        \draw (u13) -- (z1);
        \draw (u14) -- (z1);
        \draw (u15) -- (z1);
        \node[main] (y2) at (9, -0.5) {};
        \node[main] (y6) at (9, -5) {};
        \node[main] (y7) at (9, -11.5) {};
        \node[main] (w1) at (12, 0) {};
        \node[main] (w2) at (12, -1) {};
        \node[main] (w3) at (12, -2) {};
        \node[main] (w4) at (12, -3) {};
        \node[main] (w5) at (12, -4) {};
        \node[main] (w6) at (12, -5) {};
        \node[main] (w7) at (12, -6) {};
        \node[main] (w8) at (12, -7) {};
        \node[main] (w9) at (12, -8) {};
        \node[main] (w10) at (12, -9) {};
        \node[main] (w11) at (12, -10) {};
        \node[main] (w12) at (12, -11) {};
        \node[main] (w13) at (12, -12) {};
        \node[main] (w14) at (12, -13) {};
        \node[main] (w15) at (12, -14) {};
        \node[main] (z2) at (15,.-7) {};
        \draw (y2) -- (w1);
        \draw (y2) -- (w2);
        \draw (y6) -- (w3);
        \draw (y6) -- (w4);
        \draw (y6) -- (w5);
        \draw (y6) -- (w6);
        \draw (y6) -- (w7);
        \draw (y6) -- (w8);
        \draw (y7) -- (w9);
        \draw (y7) -- (w10);
        \draw (y7) -- (w11);
        \draw (y7) -- (w12);
        \draw (y7) -- (w13);
        \draw (y7) -- (w14);
        \draw (y7) -- (w15);
        \draw (w1) -- (z2);
        \draw (w2) -- (z2);
        \draw (w3) -- (z2);
        \draw (w4) -- (z2);
        \draw (w5) -- (z2);
        \draw (w6) -- (z2);
        \draw (w7) -- (z2);
        \draw (w8) -- (z2);
        \draw (w9) -- (z2);
        \draw (w10) -- (z2);
        \draw (w11) -- (z2);
        \draw (w12) -- (z2);
        \draw (w13) -- (z2);
        \draw (w14) -- (z2);
        \draw (w15) -- (z2);
    \end{tikzpicture}
    \caption{\label{fig:ga}The graph $\GG(\vc a)$ of the construction of Theorem~\ref{thm:prenpc}, for the sequence $\vc a=\langle 9, 7, 6, 5, 2, 1\rangle$; recall that this is a positive instance of TPP (without the further assumptions of Theorem~\ref{thm:tppnpc}), because $1+5+9=2+6+7=15$.}
\end{figure}

\begin{theorem}
    \label{thm:prenpc}
    Given a TPP instance $\vc a = \langle a_1, \dots, a_{3m}\rangle$, with $a_{3m} \geq 3$, $a_i-a_{i+1} \geq 3$ (for all $i=1,\dots,3m-1$) and $t =(\sum_{i=1}^{3m}a_i)/m\geq 4$, the following two statements are equivalent:
    \begin{enumerate}
        \item\label{enu:TPPone} $\vc a$ is a positive instance of Problem~\ref{prob:TPP};
        \item\label{enu:TPPtwo} $M(\vc a)$ is a DCM.
    \end{enumerate}
\end{theorem}
\begin{proof}
    \eqref{enu:TPPone} $\implies$ \eqref{enu:TPPtwo}. Let $D_1, \dots, D_m$ be the partition of $\{1,\dots,3m\}$ showing that the sequence is a positive instance. Build an undirected graph $\GG(\vc a)$ with $n$ vertices as follows:

        -- there is one vertex for every integer in the sequence $\vc a$; we will call these vertices $x_1, \dots, x_{3m}$;
       
        -- there are $t$ vertices for every class $D_j$; we will call them $y_1^j,\dots,y_t^j$ for each $j=1,\dots,m$;
        
        -- there is one vertex $z_j$ for every class $D_j$, $j=1,\dots,m$;
        
        -- for every $j=1,\dots,m$, suppose $D_j=\{i_1,i_2,i_3\}$: we add edges 
            
            \qquad  $\bullet$\ $x_{i_1} \leftrightarrow y_u^{j}$ for every $u=1,\dots, a_{i_1}$;
                
           \qquad $\bullet$\ $x_{i_2} \leftrightarrow y_u^{j}$ for every $u=a_{i_1}+1,\dots, a_{i_1}+a_{i_2}$;
               
          \qquad $\bullet$\ $x_{i_3} \leftrightarrow y_u^{j}$ for every $u=a_{i_1}+a_{i_2}+1,\dots, a_{i_1}+a_{i_2}+a_{i_3}=t$;
          
        \item for every $j=1,\dots,m$ and $u=1,\dots,t$, we add an edge $y_u^j \leftrightarrow z_j$.
 
    Each vertex $x_i$ has degree $a_i$, and each vertex $y_u^j$ has degree $2$ (because it is connected to $x_i$ for exactly one $i \in D_j$, and to $z_j$). Each $z_j$ has degree $t$.
    All vertices related to $j$ (i.e., $x_i$ for $i \in D_j$, $y_u^j$ for $u=1,\dots,t$, and $z_j$) form one connected component with $t+4$ vertices. There are $m$ such components.

    If $i \in D_j$, $x_i$ has only vertex $z_j$ at distance $2$, and $t-a_i$ vertices at distance $3$ (all the vertices of the form $y_u^j$ except those that are neighbors of $x_i$). Finally, $x_i$ has two vertices at distance three (exactly the other two vertices of the form $x_p$ for $p \in D_j$, and $p\neq i$: there are two of them, because $|D_j|=3$).
    Hence the row related to $x_i$ is of the form $[1, a_i, 1, t-a_i, 2, 0, 0, \dots]$.

    The vertices at distance two from $y_u^j$ are all the other $y_w^j$ (there are $t-1$ of them). The
    vertices at distance three are $2$ (the vertices of the form $x_p$ for $p \in D_j$, except the only one that is a direct neighbor of $y_u^j$). Hence, the row related to $y_u^j$ is of the form $[1, 2, t-1, 2, 0, 0, 0, \dots]$.

    Finally, the vertices at distance two from $z_j$ are all the vertices of the form $x_i$ for $i \in D_j$ (there are $3$ of them). Thus, the row related to $z_j$ is of the form $[1, t, 3, 0, 0, 0, \dots]$.

    It is immediate to check that the DCM of $\GG(\vc a)$ is $M(\vc a)$, meaning that $M(\vc a)$ is a DCM.

    \eqref{enu:TPPtwo} $\implies$ \eqref{enu:TPPone}. Suppose that $M(\vc a)$ is the DCM of some graph $\GG$. The graph has $n$ nodes, divided into three groups: 

        -- Group I: the first $3m$ vertices all have degree larger than 2 (because $\vc a$ is a TPP instance such that all of its elements are  $\geq 3$); group I corresponds to rows of the form $[1,a_i,1,t-a_i,2,0,\dots]$;
        
        -- Group II: the next $s$ vertices all have degree $2$; group II corresponds to rows of the form $[1,2,t-1,2,0,\dots]$;
       
        -- Group III: the last $m$ vertices all have degree $t$ (which is larger than $a_1$, hence of all $a_i$'s, because $a_i < t/2$); group III corresponds to rows of the form $[1,t,3,0,\dots]$;

    Observe that every row of $M(\vc a)$ has sum equal to $t+4$. Hence the graph has $m$ connected components, with $t+4$ vertices each. It is also worth taking note of the form of the rows of the CDCM of $\GG$:
  
        -- Group I:   $[1, a_i+1, a_i+2, t+2, t+4, t+4, \dots]$,
        
        -- Group II:  $[1, 3,     t+2,   t+4, t+4, t+4, \dots]$,

        -- Group III: $[1, t+1,   t+4,   t+4, t+4, t+4, \dots]$.
  
    Consider any component $X$ of $\GG$, and suppose that this component contains $A$ vertices of group I, $B$ vertices of group II and $C$ vertices of group III. The degree sequence of this component is (after reordering):
    \[
        [\underbrace{t,t,\dots,t}_{C}, \underbrace{a_{i_1},\dots,a_{i_A}}_A, \underbrace{2,\dots,2}_B]
    \]
    for some choice of $i_1>i_2>\dots>i_A$. Note that $A+B+C=t+4$ (the size of each component).

    Start by considering a vertex $i \in X$ of group II: it is connected to two vertices $j_1$ and $j_2$, which may each belong, in principle, to one of the three groups. 
    We shall apply Theorem~\ref{T:necc_cond2} to the row of vertex $i$ in the CDCM of $\GG$: depending on how we choose $j_1$ and $j_2$, we have the following possible intervals of integers for the third value in the row we are considering (here, we are assuming without loss of generality that $a_h>a_k$):
    \[
        \begin{tabular}{c||c|c|c}
            \diagbox{$j_1$ (degree)}{$j_2$ (degree)} & Group I ($a_k$) & Group II ($2$) & Group III ($t$) \\
            \hline
            Group I ($a_h$)  & $a_h+1 \dots a_h+a_k+2$ & $a_h+1 \dots a_h+4$ & $t+1 \dots a_h+t+2$ \\
            \hline
            Group II ($2$) & & $3\dots 6$ & $t+1 \dots t+4$ \\
            \hline
            Group III ($t$) & & & $t+1 \dots 2t+2$ \\
        \end{tabular}
    \]
    The third value in the row of the CDCM for group II is $t+2$, and the only intervals containing $t+2$ are those in the last column (since $a_{3m}\geq 3$, $t$ is certainly larger than 6). This means that every vertex of group II is connected to at least one vertex of group III.
    As a consequence, $B>0$ implies $C>0$.
    
    \smallskip
    On the other hand, applying again Theorem~\ref{T:necc_cond2}, no vertex $i \in X$ of group I can be connected to a node of group III, because $t+1>a_i+2$. Also, no two vertices of group I can be connected to each other: suppose, by contradiction that a vertex of degree $a_h$ is connected to a vertex of degree $a_k$ with $a_h>a_k$. Then we should have (by Theorem~\ref{T:necc_cond2}) $a_h+1 \leq a_k+2$, that is $a_h-a_k \leq 2$, which is impossible because $a_h-a_k \geq 3$. 
    We conclude that vertices of group I are connected only to vertices of group II. As a consequence, $A>0$ implies $B>0$.
    
    For the reasons above, every component of $\GG$ must contain at least one vertex of group III. Since there are $m$ components and exactly $m$ vertices of group III, we conclude that each component contains \emph{exactly one} vertex of group III, that is, $C=1$.

    The $A$ vertices of group I in component $X$ are all connected to vertices of group II (this is, as we discussed above, the only possibility for them), and those vertices of group II must be distinct (they have degree 2, and each of them is connected to at least one vertex of group III); so there must  be at least $a_{i_1}+\dots+a_{i_A}$ vertices of group II in the component, say $a_{i_1}+\dots+a_{i_A}+e_X$ for some $e_X\geq 0$. Assume for the moment that $e_X=0$.

    The only vertex of group III is itself connected to $t$ vertices of group II, so we conclude that $a_{i_1}+\dots+a_{i_A}=t$ and $B=t$. Since $A+B+C=t+4$, we have $A+t+1=t+4$, hence $A=3$.  

    These conditions, together with the fact that this property holds for all components of $\GG$, allows us to conclude that $\vc a$ is a positive instance of TPP.

    Finally, note that $e_X$ \emph{must} be zero, because if even one component contained some extra vertex of group II, then altogether there would be  more than $s$ vertices of group II in the graph, which is impossible.\qed
\end{proof}

\begin{corollary}
    Recognizing whether a given matrix is a (C)DCM is strongly NP-complete. 
\end{corollary}
\begin{proof}
    Theorem~\ref{thm:prenpc} reduces TPP to recognizing DCMs, hence the result follows by Theorem~\ref{thm:tppnpc}. The fact that TPP is strongly NP-complete is needed for the construction to be admissible (the size of the matrix is pseudopolynomial in the input sequence $\vc a$).\qed
\end{proof}

\section{Conclusions and Future Work}

As mentioned, our negative result is relevant especially in the area of axiomatic centrality, when one tries to build counterexamples for specific properties of geometric centralities. On the other hand, there are families of graphs for which the DCM decision problem is feasible (e.g., out-directed trees), and there are many operations that preserve this property (e.g., Kronecker graph product): a natural area for future research will be to find large families of graphs for which the decision problem can be solved in polynomial time.

{\footnotesize
    \section*{\small Acknowledgements}
    
    This work was supported in part by project SERICS (PE00000014) under the NRRP MUR program funded by the EU - NGEU. Views and opinions expressed are those of the authors, and do not necessarily reflect those of the European Union or the Italian MUR. Neither the European Union nor the Italian MUR can be held responsible for them.
}

\bibliographystyle{plain}
\bibliography{biblio,ian,extra}

\end{document}